\documentclass[11pt]{article}
\usepackage{fullpage}
\usepackage{authblk}
\usepackage{graphicx}
\usepackage{amsmath}
\usepackage{amsfonts}
\usepackage{amssymb}
\usepackage{amsthm}
\usepackage{mathrsfs}
\usepackage{bm}
\usepackage{mathtools}
\usepackage{enumitem}

\usepackage{url}
\usepackage{hyperref}

\usepackage{caption}
\usepackage{subcaption}
\newtheorem{theorem}{Theorem}
\newtheorem{lemma}[theorem]{Lemma}

\newtheorem{corollary}[theorem]{Corollary}

\newtheorem*{claim*}{Claim}
\newtheorem*{remark*}{Remark}

\newtheorem*{definition*}{Definition}

\makeatletter
\renewcommand\section{%
  \@startsection{section}{1}
                {\z@}%
                {-3.5ex \@plus -1ex \@minus -.2ex}%
                {2.3ex \@plus.2ex}%
                {\large\bfseries}
}
\renewcommand\subsection{%
  \@startsection{subsection}{2}
                {\z@}%
                {-3.25ex\@plus -1ex \@minus -.2ex}%
                {1sp}
                {\normalsize\bfseries}
}
\renewcommand\subsubsection{%
  \@startsection{subsubsection}{3}
                {\z@}%
                {-3.25ex\@plus -1ex \@minus -.2ex}%
                {1sp}
                {\normalfont\normalsize}
}
\makeatother

\title{{\LARGE\bf  {O}n the Convexity of Independent Set Games\thanks{This work was supported in part by National Natural Science Foundation of China (11871442, 11971447) and Fundamental Research Funds for the Central Universities (201713051, 201964006).}}}

\author{Han Xiao\thanks{Corresponding author. Email: hxiao@ouc.edu.cn.},~~Yuanxi Wang,~~Qizhi Fang}

\affil{School of Mathematical Sciences\\Ocean University of China\\Qingdao, China}
\date{}
\begin{document}


\maketitle

\openup 1.2\jot


\begin{abstract}
Independent set games are cooperative games defined on graphs,
where players are edges and the value of a coalition is the maximum cardinality of independent sets in the subgraph defined by the coalition.
In this paper, we investigate the convexity of independent set games, as convex games possess many nice properties both economically and computationally.
For independent set games introduced by Deng et al.\,\cite{DengIbar99},
we provide a necessary and sufficient characterization for the convexity, i.e., every non-pendant edge is incident to a pendant edge in the underlying graph.
Our characterization immediately yields a polynomial time algorithm for recognizing convex instances of independent set games.
Besides, we introduce a new class of independent set games and provide an efficient characterization for the convexity.

\hfill


\hfill

\hfill
\end{abstract}

\section{Introduction}

Cooperative games form an important class in game theory, which have a lot of applications in economics, computer science and mathematics.
One major problem in cooperative games is to maximize the profit by cooperation and allocate the profit among participants in the cooperation.
There are many criteria for evaluating how ``good''  an allocation is, such as fairness and stability.
Emphases on different criteria lead to different solution concepts, e.g., the core, the Shapley value, the nucleolus, the bargaining set, and the von Neumann-Morgenstern solution  \cite{Shub81}.
Among those solution concepts, the core which addresses the issue of stability is one of the most attractive solution concepts.
The core is the set of allocations where no coalition has an incentive to split off from the grand coalition, and does better on its own.
Hence a game with a non-empty core is especially interesting.
A special subclass of games with a non-empty core is formed by convex \footnote{concave for cooperative games minimizing the cost by cooperation and allocating the cost among participants in the cooperation.} games.

Convex games were introduced by Shapley \cite{Shap71}, which exhibit many desirable properties in cooperative game theory.
In particular,
(i) the core is always non-empty and a core allocation can be found in polynomial time \cite{Shap71};
(ii) testing whether an allocation belongs to the core can be performed in polynomial time \cite{GLS93};
(iii) computing the nucleolus can be done in polynomial time \cite{Kuip96};
(iv) there is an appealing snowball effect, i.e., the incentive to join a coalition increases as the coalition grows larger  \cite{Shap71}.
We refer to \cite{MascPele71, Shap71} for many interesting properties of convex games involving other solution concepts.
Hence the convexity of cooperative games has attracted a lot of research efforts,
especially for cooperative games arising from combinatorial optimization problems.
However, only a few combinatorial optimization games are universally convex \cite{AumaMasc85,CuriPede89,HameBorm95,LittOwen73}.
Hence one working direction is to characterize the condition for convexity.
There is a line of research where convexity/concavity of cooperative games is characterized by the property of underlying graphs.
Van den Nouweland and Borm \cite{vanBorm91} showed that communication vertex games are convex if and only if the underlying graph is cycle-complete and communication arc games are convex if and only if the underlying graph is cycle-free.
Herer and Penn \cite{HerePenn95} showed that Steiner traveling salesman games are concave if the underlying graph is a $1$-sum of $K_4$ and outerplanar graphs.
Hamers \cite{Hame97} showed that Chinese postman games are concave if the underlying graph is weakly cyclic.
Okamoto \cite{Okam03} showed that vertex cover games are concave if and only if the underlying graph is $(K_3,P_4)$-free, and coloring games are concave if and only if the underlying graph is complete multipartite.
Based on the result of Hamers \cite{Hame97}, Albizuri and Hamers \cite{AlbiHame14} characterized the concavity of some variants of Chinese postman games.
Kobayashi and Okamoto \cite{KobaOkam14} initialized the study for the concavity of spanning tree games,
where a sufficient condition and a necessary condition were given separately.
Koh and Sanit\`{a} \cite{KohSani20} provided the first necessary and sufficient characterization for the concavity of spanning tree games.
Platz \cite{Plat19} gave a complete characterization for the concavity of multi-depot Steiner traveling salesman games.

In this paper, we focus on the convexity of independent set games.
Independent set games were introduced by Deng, Ibaraki and Nagamochi \cite{DengIbar99}, which model the following scenario with projects and participants.
Every participant is suitable for two projects but only allowed to join one project.
Every project requires \emph{all} suitable participants to cooperate to be done.
The problem of maximizing doable projects can be viewed as a maximum independent set problem where projects are vertices and every participant is an edge joining two suitable projects.
We provide a necessary and sufficient characterization for the convexity of independent set games via the underlying graph. 
Our characterization implies that convex independent set games can be recognized effciently.
Besides, we introduce a new class of independent set games and provide a necessary and sufficent characterization for the convexity.

The rest of this paper is structured as follows.
In Section \ref{sec:preliminaries}, notions and notations in graph theory and game theory are reviewed.
Section \ref{sec:CVX_ISGame} is devoted to an efficient characterization for the convexity of independent set games.
In Section \ref{sec:CVX_RISGame}, we introduce a new class of independent set games and characterize the convexity.
Section \ref{sect:conclusions} concludes the results in this paper and discusses the directions of future work.

\section{Preliminaries}
\label{sec:preliminaries}

\subsection{Graphs}
We assume that the readers have a moderate familiarity with graphs.
However, some notions and notations used in this paper should be clarified before proceeding.
Throughout, a graph is always finite, undirected and simple.
For $n\in \mathbb{N}$,
we use $K_n$ to denote the complete graph with $n$ vertices,
use $K_{1,n}$ to denote the graph which is a \emph{star}, i.e., a complete bipartite graph where one part has one vertex and the other part has $n$ vertices,
use $C_n$ to denote the graph which is a cycle with $n$ vertices,
and use $P_n$ to denote the graph which is a path with $n$ vertices.
Since $K_2$ is isomorphic to $K_{1,1}$ and $P_3$ is isomorphic to $K_{1,2}$, both $K_2$ and $P_3$ are stars.
Let $H$ be a graph.
We use $V(H)$ to denote the vertex set of $H$ and use $E(H)$ to denote the edge set of $H$.
A graph is said $H$-\emph{free} if it contains no subgraph isomorphic to $H$.
Let $G=(V,E)$ be a graph.
For any $v\in V$,
$N_G(v)$ denotes the set of vertices \emph{adjacent} to $v$,
$\delta_G(v)$ denotes the set of edges \emph{incident} to $v$,
and $d_G(v)$ denotes the degree of $v$.
A vertex is \emph{isolated} if it is a vertex with degree zero.
A vertex is \emph{pendant} if it is a vertex with degree one.
An edge is \emph{pendant} if it is incident to a pendant vertex.
For any $U\subseteq V$, $G[U]$ denotes the induced subgraph of $G$.
In particular, $G[\emptyset]$ is an empty graph which has no vertex.
For any $F\subseteq E$,
$V\langle F\rangle $ denotes the set of vertices incident \emph{only} to edges in $F$,
and $G[F]$ denotes the edge-induced subgraph of $G$.
An \emph{independent set} of $G$ is a vertex set $U\subseteq V$ such that $G[U]$ has no edge. 
An \emph{edge cover} of $G$ is an edge set $F\subseteq E$ such that every vertex of $G$ is incident to an edge in $F$.
For simplicity, we use $\alpha(G)$ to denote the maximum cardinality of independent sets in $G$.

\subsection{Cooperative games}
Let $\Gamma=(N,\gamma)$ be a cooperative game,
where $N$ is the set of players and $\gamma:2^N\rightarrow \mathbb{R}$ is the characterizatic function with $\gamma(\emptyset)=0$.
A subset $S$ of $N$ is called a \emph{coalition} and $N$ is called the \emph{grand coalition}.
For each coalition $S$,
$\gamma(S)$ represents the value distributed among the players in $S$.
We call $\Gamma$ \emph{convex} if for any $S,T\subseteq N$,
\begin{equation}
\label{eq:CVX1}
\gamma(S)+\gamma(T)\leq \gamma(S\cap T)+\gamma(S\cup T),
\end{equation}
or equivalently, for any $i\in N$ and any $S\subseteq T\subseteq N\backslash \{i\}$,
\begin{equation}
\label{eq:CVX2}
\gamma(S\cup \{i\})-\gamma(S)\leq \gamma(T\cup \{i\})-\gamma(T).
\end{equation}
We call $\Gamma$ \emph{concave} if the reverse inequality holds in $(\ref{eq:CVX1})$ or $(\ref{eq:CVX2})$,
and \emph{additive} if the equality holds in $(\ref{eq:CVX1})$ or $(\ref{eq:CVX2})$.

An \emph{allocation} of $\Gamma$ is a vector $\boldsymbol{x}\in \mathbb{R}^{N}$ which consists of proposed amounts to be shared by players in $N$.
An allocation $\boldsymbol{x}$ is called \emph{efficient} if $\sum_{i\in N}x_i=\gamma(N)$, and called \emph{coalitionally rational} if $\sum_{i\in S}x_i
\geq \gamma(S)$ for any $S\subseteq N$.
The \emph{core} of $\Gamma$ is the set of allocations that are efficient and coalitionally rational.
The core of a cooperative game may be empty.
A cooperative game is called \emph{balanced} if the core is non-empty.
Balanced games contain convex games as a special subclass \cite{Shap71}.

\section{Convexity of independent set games}
\label{sec:CVX_ISGame}

The \emph{independent set game} introduced by Deng et al. \cite{DengIbar99} is a cooperative game $\Gamma_G=(E,\gamma)$ defined on a graph $G=(V,E)$, where $E$ is the set of players and $\gamma:2^{E}\rightarrow \mathbb{N}$ is the characteristic function such that $\gamma(F)=\alpha(G[V\langle F\rangle])$ for any $F\subseteq E$.
We always assume that the underlying graph of an independent set game has no isolated vertex.
Deng et al. \cite{DengIbar99} showed that an independent set game is balanced if and only if the underlying graph is a K\"{o}nig-Egerv\'{a}ry graph \cite{Demi79, Schr03}, where the maximum cardinality of independent sets is equal to the minimum cardinality of edge covers.
In this paper, we show that an independent set game is convex if and only if every non-pendant edge is incident to a pendant edge in the underlying graph.
We remark that (non-)pendant vertices and edges always refer to vertices and edges in the ground graph throughout this section.

Let $G=(V,E)$ be a graph without isolated vertices.
Let $e\in E$ be an edge with endpoints $u_1, u_2$ and $F\subseteq E\backslash \{e\}$ be a set of edges.
Clearly, neither $u_1$ nor $u_2$ belongs to $V\langle F \rangle$.
Let $F'=F\cup \{e\}$.
It follows that $V\langle F \rangle \subseteq V\langle F' \rangle$.
Thus every independent set of $G[V\langle F\rangle]$ is also an independent set of $G[V\langle F'\rangle]$.
Notice that $V\langle F' \rangle\subseteq V\langle F\rangle \cup\{u_1,u_2\}$.
It follows that a maximum independent set of $G[V\langle F'\rangle]$ is also a maximum independent set of $G[V\langle F\rangle]$ if it is a subset of $V\langle F\rangle$.
Let $\Gamma_G$ be the independent set game on $G$.
It follows that $\gamma(F')\geq \gamma(F)$.
In the following, we distinguish two cases of $e$.

\begin{lemma}
\label{lem:1}
If $e$ is a pendant edge in $G$, then $\gamma(F')=\gamma(F)+1$.
\end{lemma}
\begin{proof}
Let $I_F$ and $I_{F'}$ be a maximum independent set in $G[V\langle F\rangle]$ and $G[V\langle F'\rangle]$, respectively.
We show that $\lvert I_{F'}\rvert=\lvert I_F\rvert+1$.
We may assume that $e$ is only incident to one pendant vertex, say $u_1$, since otherwise $\lvert I_{F'}\rvert=\lvert I_F\rvert+1$ is trivial.
It follows that $u_1\in V\langle F'\rangle\backslash V\langle F\rangle$.

We claim that $u_1\in I_{F'}$.
To see this, we distinguish two cases of $V\langle F'\rangle$.
It is trivial if $V\langle F'\rangle=V\langle F\rangle \cup \{u_1\}$, as in this case $u_1$ is an isolated vertex in $G[V\langle F'\rangle]$.
Hence assume that $V\langle F'\rangle=V\langle F\rangle \cup \{u_1,u_2\}$.
In this case we have $I_{F'}\cap \{u_1,u_2\}\not=\emptyset$, since otherwise $I_{F'}\cup \{u_1\}$ is an independent set in $G[V\langle F'\rangle]$, which contradicts the maximality of $I_{F'}$.
Since $u_1$ is a pendant vertex, we may always assume that $u_1\in I_{F'}$ by replacing $u_2$ with $u_1$ in $I_{F'}$ when necessary.

Now on one hand, $I_{F'}\backslash \{u_1\}$ is an independent set in $G[V\langle F\rangle]$. 
Thus $\lvert I_{F'}\backslash \{u_1\}\rvert = \lvert I_{F'}\rvert -1\leq \lvert I_F\rvert$.
On the other hand, since $u_1$ is a pendant vertex, $I_F\cup \{u_1\}$ is an independent set in $G[V\langle F'\rangle]$.
Thus $\lvert I_{F'}\rvert\geq \lvert I_F\cup \{u_1\}\rvert=\lvert I_F\rvert +1$.
It follows that $\lvert I_{F'}\rvert=\lvert I_F\rvert+1$.
\end{proof}

\begin{lemma}
\label{lem:2}
If $e$ is not a pendant edge in $G$, then $\gamma(F')= \gamma(F)$ or $\gamma(F')= \gamma(F)+1$.
Moreover, if $\gamma(F')=\gamma (F)+1$, then every maximum independent set in $G[V\langle F'\rangle]$ contains an endpoint of $e$ which is not adjacent to any pendant vertex.
\end{lemma}
\begin{proof}
Assume that $e$ is not a pendant edge in $G$.
It is trivial that $\gamma(F)\leq \gamma(F')$.
We show that $\gamma(F')\leq \gamma(F)+1$.
Without loss of generality, assume that $V\langle F'\rangle \not= V\langle F\rangle$,
since otherwise we have $\gamma(F')=\gamma(F)$.
Let $I_F$ and $I_{F'}$ be a maximum independent set in $G[V\langle F\rangle]$ and $G[V\langle F'\rangle]$, respectively.
It suffices to show that $\lvert I_{F'}\rvert\leq \lvert I_F\rvert+1$.
Notice that $\lvert I_{F'}\rvert= \lvert I_F\rvert$ if $I_{F'}\cap \{u_1,u_2\}=\emptyset$, as $I_{F'}$ is also a maximum independent set in $G[V\langle F\rangle]$.
Hence, without loss of generality, assume that $u_1\in I_{F'}$.
Notice that $I_{F'}\backslash \{u_1\}$ is an independent set in $G[V\langle F\rangle]$.
It follows that $\lvert I_{F}\rvert\geq \lvert I_{F'}\backslash \{u_1\}\rvert=\lvert I_{F'}\rvert -1$.
Therefore, we have $\gamma (F)\leq \gamma(F')\leq \gamma(F)+1$.

Now further assume $\gamma(F')=\gamma (F)+1$.
Recall that $u_1\in I_{F'}$ is the endpoint of $e$.
We show that $u_1$ is not adjacent to any pendant vertex by proving that $d_G (v)\geq 2$ for any $v\in N_G (u_1)$.
Assume to the contrary that there is a vertex $v^*\in N_G (u_1)$ with $d_G (v^*)=1$.
Then $\{u_1,v^*\}$ is a pendant edge, implying that $v^*\not=u_2$.
It follows that $\delta_G (u_1)\subseteq F'$ and $\delta_G (u_1)\backslash \{e\}\subseteq F$.
Thus $\{u_1,v^*\}$ is an edge in $F$.
Since $v^*$ is a pendant vertex, we have $v^*\in V\langle F\rangle$ and $u_1\not \in V\langle F\rangle$, implying that $v^*$ is an isolated vertex in $G[V\langle F\rangle]$.
Let $I'_{F'}=I_{F'}\backslash \{u_1\}$.
Notice that $I'_{F'}$ is an independent set in $G[V\langle F\rangle]$ with cardinality $\lvert I_{F'}\rvert-1$.
Since $v^*$ and $u_1$ are adjacent in $G[V\langle F'\rangle ]$, we have $v^*\not\in I_{F'}$.
Then $I'_{F'}\cup \{v^*\}$ is an independent set in $G[V\langle F\rangle]$ with cardinality $\lvert I_{F'}\rvert$, which contradicts $\gamma(F')=\gamma(F)+1$.
\end{proof}

Now we are ready to present our main result.

\begin{theorem}
\label{thm:CVX_ISGame}
Let $G=(V,E)$ be a graph without isolated vertices and $\Gamma_G=(E,\gamma)$ be the independent set game on $G$.
Then $\Gamma_G$ is convex if and only if every non-pendant edge is incident to a pendant edge in $G$.
\end{theorem}

\begin{proof}
We first prove the ``only if'' part.
Assume to the contrary that there is a non-pendant edge $e\in E$ which is not incident to any pendant edge in $G$.
Let $u_1,u_2$ be the endpoints of $e$.
Thus $d_G (v)\geq 2$ for any $v\in N_G (u_1)\cup N_G (u_2)$.
Let $W_{u_1 u_2}$ be the set of vertices from $N_G(u_1)\cap N_G(u_2)$ with degree two, i.e., $W_{u_1 u_2}=\{w\in N_G(u_1)\cap N_G(u_2): d_G(w)=2\}$.
In the following, we distinguish two cases of $W_{u_1 u_2}$ and show that either case leads to a contradiction.
Hence we conclude that every non-pendant edge is incident to a pendant edge in $G$.

Case $1$: $\lvert W_{u_1 u_2}\rvert\leq 1$.
Let $S=\delta_G (u_1)$ and $T=\delta_G (u_2)$.
Since $d_G (v)\geq 2$ for any $v\in N_G (u_1)\cup N_G (u_2)$,
we have $\gamma(S)=\gamma(T)=\alpha (K_1)=1$ and $\gamma(S\cap T)=\gamma(\{e\})=\alpha(K_0)=0$.
Moreover, $\gamma(S\cup T)=\alpha(K_2)=1$ if $\lvert W_{u_1 u_2}\rvert=0$, and $\gamma(S\cup T)=\alpha(K_3)=1$ if $\lvert W_{u_1 u_2} \rvert=1$.
It follows that $\gamma(S)+\gamma(T)>\gamma(S\cap T)+\gamma(S\cup T)$, which contradicts the convexity of $\Gamma_G$.

Case $2$: $\lvert W_{u_1 u_2}\rvert\geq 2$.
Let $w^*\in W_{u_1 u_2}$.
Clearly, $N_G(w^*)=\{u_1,u_2\}$.
The following arguments hold for $i=1,2$.
Since $d_G (v)\geq 2$ for any $v\in N_G (u_i)\cup N_G(w^*)$, $\{u_i,w^*\}$ is neither a pendant edge nor incident to any pendant edge.
Let $W_{u_i w^*}$ be the vertex set defined analogously to $W_{u_1 u_2}$,
i.e., $W_{u_i w^*}=\{w\in N_G(u_i)\cap N_G(w^*): d_G(w)=2\}$.
Instead of $W_{u_1 u_2}$, we turn to consider the cardinality of $W_{u_i w^*}$.
Clearly, $W_{u_i w^*}\subseteq N_G(w^*)$.
It follows that $W_{u_i w^*}\subseteq \{u_1,u_2\}$.
However, $\lvert W_{u_1 u_2}\rvert\geq 2$ implies that $d_G(u_i)\geq 3$.
It follows that $W_{u_i w^*}=\emptyset$ which boils down to Case $1$.

Now we prove the ``if'' part.
Assume to the contrary that $\Gamma_G$ is not convex.
Then there is an edge $e\in E$ and two edge sets $S\subseteq T\subseteq E\backslash \{e\}$ such that 
\begin{equation}
\label{eq:NonCVX}
\gamma(S')-\gamma(S)>\gamma(T')-\gamma(T),
\end{equation}
where $S'=S\cup \{e\}$ and $T'=T\cup \{e\}$.
Let $u_1,u_2$ be the endpoints of $e$.
By Lemma \ref{lem:1}, $e$ is not a pendant edge in $G$, since otherwise $\gamma(S')-\gamma(S)=\gamma(T')-\gamma(T)=1$.
Moreover, Lemma \ref{lem:2} implies that
\begin{equation}
\label{eq:Equality1}
\gamma(S')=\gamma(S)+1
\end{equation}
and
\begin{equation}
\label{eq:Equality2}
\gamma(T')=\gamma(T).
\end{equation}

Let $I_{S'}$ and $I_{T}$ be a maximum independent set in $G[V\langle S'\rangle]$ and $G[V\langle T\rangle]$, respectively.
By Lemma \ref{lem:2}, \eqref{eq:Equality1} implies that an endpoint of $e$, say $u_1$, belongs to $I_{S'}$ but is not adjacent to any pendant vertex.
It follows that $\delta_G (u_1)\subseteq S'\subseteq T'$.
Since $u_1\not\in V\langle T\rangle$, we have $u_1\not\in I_T$.
In the following, we construct an independent set in $G[V\langle T'\rangle]$ with $u_1$ from $I_T$, the size of which is larger than $I_T$.
However, this contradicts \eqref{eq:Equality2}.
Hence we conclude that $\Gamma_G$ is convex.

It remains to construct an asserted independent set in $G[V\langle T'\rangle]$.
We first show that $I_T\cap N_G(u_1)\not=\emptyset$ and every vertex in $I_T\cap N_G(u_1)$ is adjacent to a pendant vertex from $V\langle T\rangle$.
Notice that $I_T$ is also an independent set in $G[V\langle T'\rangle]$.
If $I_T\cap N_G(u_1)=\emptyset$, i.e., $u_1$ is not adjacent to any vertex in $I_T$, then $I_T\cup \{u_1\}$ is an independent set in $G[V\langle T'\rangle]$ with cardinality $\lvert I_T\rvert+1$, which is trivial.
Hence $I_T\cap N_G(u_1)\not=\emptyset$.
Since $u_1$ is not adjacent to any pendant vertex, every vertex in $I_T\cap N_G(u_1)$ is adjacent to a pendant vertex.
Notice that $I_T\cap N_G(u_1)\subseteq V\langle T\rangle$.
Hence every pendant vertex which is adjacent to a vertex in $I_T\cap N_G(u_1)$ also belongs to $V\langle T\rangle$.
Now we construct a larger independent set in $G[V\langle T'\rangle]$ with $u_1$ from $I_T$.
Let $I'_T$ be the vertex set obtained from $I_T$ by replacing each vertex in $I_T\cap N_G(u_1)$ with a pendant vertex from $V\langle T\rangle$ adjacent to it.
Clearly, $I'_T\subseteq V\langle T\rangle$.
Moreover, $I'_T$ is a maximum independent set in $G[V\langle T\rangle]$.
Hence $I'_T$ is also an independent set in $G[V\langle T'\rangle]$.
Since $u_1$ is not adjacent to any vertex in $I'_T$, $I'_T\cup \{u_1\}$ is an independent set in $G[V\langle T'\rangle]$ with cardinality $\lvert I_T\rvert+1$.
\end{proof}

Theorem \ref{thm:CVX_ISGame} can be strengthened to a characterization for the additivity of independent set games.
\begin{corollary}
\label{thm:ADD_ISGame}
Let $G=(V,E)$ be a graph without isolated vertices and $\Gamma_G=(E,\gamma)$ be the independent set game on $G$.
Then $\Gamma_G$ is additive if and only if every non-pendant vertex is adjacent to a pendant vertex in $G$.
\end{corollary}
\begin{proof}
Notice that if every non-pendant vertex is adjacent to a pendant vertex, then every non-pendant edge is incident to a pendant edge. 
Hence to prove the additivity of $\Gamma_G$, it suffices to show that $\Gamma_G$ is concave if and only if every non-pendant vertex is adjacent to a pendant vertex in $G$.

We first prove the ``only if'' part.
Assume to the contrary that there is a non-pendant vertex $v^*\in V$ which is not adjacent to any pendant vertex in $G$.
Let $S,T\subseteq \delta_G(v^*)$ be two non-empty edge sets such that $S\cap T=\emptyset$ and $S\cup T=\delta_G (v^*)$.
It follows that $\gamma(S)=\gamma(T)=\alpha(K_0)=0$, $\gamma(S\cap T)=\gamma(\emptyset)=0$
and $\gamma(S\cup T)=\alpha(K_1)=1$.
Therefore, we have $\gamma(S)+\gamma(T)<\gamma(S\cap T)+\gamma(S\cup T)$,
which contradicts the concavity of $\Gamma_G$.

Now we prove the ``if'' part.
Assume to the contrary that $\Gamma_G$ is not concave.
Then there is an edge $e$ and two edge sets $S\subseteq T\subseteq E\backslash \{e\}$ such that 
\begin{equation}
\label{eq:NonCCV}
\gamma(S')-\gamma(S)<\gamma(T')-\gamma(T),
\end{equation}
where $S'=S\cup \{e\}$ and $T'=T\cup \{e\}$.
By Lemma \ref{lem:1}, $e$ is not a pendant edge in $G$, since otherwise $\gamma(S')-\gamma(S)=\gamma(T')-\gamma(T)=1$.
By Lemma \ref{lem:2}, we have
\begin{equation}
\label{eq:Equality3}
\gamma(T')=\gamma(T)+1
\end{equation}
and
\begin{equation}
\label{eq:Equality4}
\gamma(S')=\gamma(S).
\end{equation}
Moreover, \eqref{eq:Equality3} implies that an endpoint of $e$ is neither a pendant vertex nor adjacent to any pendant vertex, which contradicts our basic assumption.
\end{proof}

One natural question is that whether there are graphs inducing convex but not additive independent set games.
In Figure \ref{fig:example}, we enumerate all connected graphs with $5$ vertices that induce convex independent set games.
Clearly, $G_1$ and $G_2$ also induce additive independent set games, but $G_3$ and $G_4$ do not induce additive independent set games.

\vspace{-2em}
\begin{figure}[h!]
\hspace{-20.4em}\includegraphics[width=1.95\textwidth]{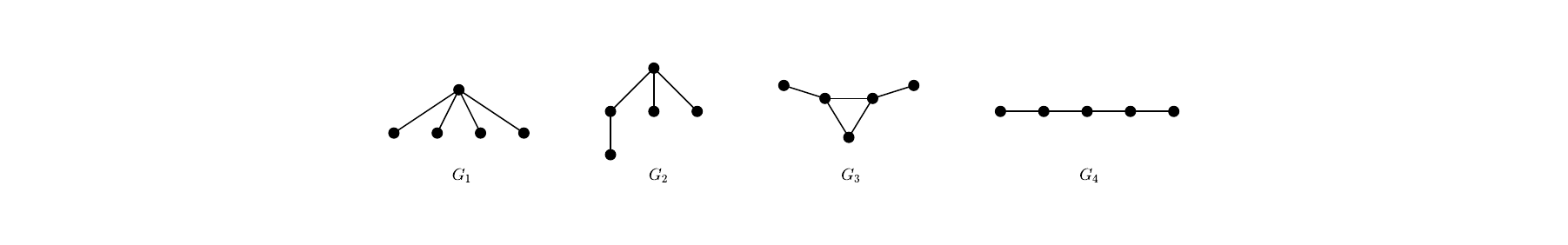}
\vspace{-4em}
\caption{All connected graphs with $5$ vertices that induce convex independent set games.}
\label{fig:example}
\end{figure}

Since pendant edges and pendant vertices can be recognized in polynomial time by checking the degree of every vertex,
Theorem \ref{thm:CVX_ISGame} and Corollary \ref{thm:ADD_ISGame} suggest that both convex instances and additive instances of independent set games can be recognized efficiently. 
\begin{corollary}
\label{thm:Complexity_ISGame}
  Let $G=(V,E)$ be a graph without isolated vertices and $\Gamma_G=(E,\gamma)$ be the independent set game on $G$.
  Then both convexity and additivity of $\Gamma_G$ can be determined in polynomial time.
\end{corollary}

\section{Convexity of relaxed independent set games}
\label{sec:CVX_RISGame}
Independent set games introduced by Deng et al. \cite{DengIbar99} have limitation in the sense that every project requires \emph{all} suitable participants to cooperate to be done.
Hence a project fails if it misses any suitable participant in a coalition.
However, a project might only require parts of suitable participants to cooperate to be done, not necessarily all suitable participants.
Here we consider a special case of this scenario, where every project requires all suitable participants in the current coalition, not necessarily all suitable participants in the grand coalition, to cooperate to be done.
Independent set games complying with the new requirement are called \emph{relaxed independent set games}.
Formally, the relaxed independent set game is a cooperative game $\hat{\Gamma}_G=(E,\hat{\gamma})$ defined on a graph $G=(V,E)$, where $E$ is the set of players and $\hat{\gamma}:2^{E}\rightarrow \mathbb{N}$ is the characteristic function such that $\hat{\gamma}(F)=\alpha(G[F])$ for any $F\subseteq E$.
We show that the convexity of relaxed independent set games admits an efficient characterization.

\begin{theorem}
\label{thm:CVX_RISGame}
Let $G=(V,E)$ be a graph without isolated vertices and $\hat{\Gamma}_G=(E,\hat{\gamma})$ be the relaxed independent set game on $G$.
Then $\hat{\Gamma}_G$ is convex if and only if $G$ is $(K_3,P_4)$-free.
\end{theorem}

\begin{proof}
We first prove the ``only if'' part.
Assume to the contrary that $G$ is not $(K_3,P_4)$-free.
We show that any subgraph isomorphic to $K_3$ or $P_4$ gives rise to non-convex instances of $\hat{\Gamma}_G$.

Assume that $H$ is either $K_3$ or $P_4$.
Let $S\subseteq E(H)$ be a set consisting of two incident edges in $H$,
and let $T= E(H)\backslash S$.
Notice that $\hat{\gamma}(S)=\alpha(P_3)=2$, $\hat{\gamma}(T)=\alpha(K_2)=1$, and $\hat{\gamma}(S\cap T)=\hat{\gamma}(\emptyset)=0$.
Moreover, $\hat{\gamma}(S\cup T)=\alpha(K_3)=1$ if $H=K_3$ and $\hat{\gamma}(S\cup T)=\alpha(P_4)=2$ if $H=P_4$.
In either case, we have $\hat{\gamma}(S)+\hat{\gamma}(T)>\hat{\gamma}(S\cap T)+\hat{\gamma}(S\cup T)$,
which contradicts the convexity of $\hat{\Gamma}_G$.

Now we prove the ``if'' part.
Let $H_1,\ldots,H_r$ be the components of $G$.
Okamoto \cite{Okam03} showed that a graph is $(K_3,P_4)$-free if and only if every component is a star.
It follows that 
 \begin{equation}
 \alpha (G)=\sum_{i=1}^r \alpha(H_i)=\sum_{i=1}^r \lvert E(H_i)\rvert=\lvert E\rvert.
 \end{equation}
Furthermore, we have
\begin{equation}
\hat{\gamma}(F)=\alpha (G[F])=\sum_{i=1}^r \lvert E(H_i)\cap F\rvert=\lvert F\rvert
\end{equation}
for any $F\subseteq E$.
Let $e\in E$ and $S\subseteq T\subseteq E\backslash \{e\}$.
It follows that
\begin{equation}
\hat{\gamma}(S')-\hat{\gamma}(S)=\lvert S'\rvert-\lvert S\rvert=1
\end{equation}
and
\begin{equation}
\hat{\gamma}(T')-\hat{\gamma}(T)=\lvert T'\rvert-\lvert T\rvert=1,
\end{equation}
where $S'=S\cup \{e\}$ and $T'=T\cup \{e\}$.
Therefore, $\hat{\Gamma}_G$ is additive and hence convex.
\end{proof}
From the proof above, Theorem \ref{thm:CVX_RISGame} can be strengthened as follows.
\begin{corollary}
\label{thm:ADD_RISGame}
Let $G=(V,E)$ be a graph without isolated vertices and $\hat{\Gamma}_G=(E,\hat{\gamma})$ be the relaxed independent set game on $G$.
Then $\hat{\Gamma}_G$ is additive if and only if $G$ is \emph{(}$K_3$,$P_4$\emph{)}-free.
\end{corollary}

Since $(K_3,P_4)$-free graphs can be recognized in polynomial time,
Theorem \ref{thm:CVX_RISGame} and Corollary \ref{thm:ADD_RISGame} suggest that convex (actually additive) instances of relaxed independent set games can be recognized efficiently.

\begin{corollary}
\label{thm:Complexity_RISGame}
  Let $G=(V,E)$ be a graph without isolated vertices and $\hat{\Gamma}_G=(E,\hat{\gamma})$ be the relaxed independent set game on $G$.
  Then convexity (actually additivity) of $\hat{\Gamma}_G$ can be determined in polynomial time.
\end{corollary}

\section{Concluding remarks}
\label{sect:conclusions}

We study the convexity of independent set games in this paper.
We show that an independent set game is convex if and only if every non-pendant edge is incident to a pendant edge in the underlying graph.
Since convex games are population monotonic and hence totally balanced, a possible direction for future work is to characterize population monotonicity or total balancedness of independent set games based on the characterization of convexity.
Deng et al. \cite{DengIbar99} proved that an independent set game is balanced if and only if the underlying graph is a K\"{o}nig-Egerv\'{a}ry graph.
Thus any characterization for convexity, population monotonicity or total balancedness of independent set games also provides a sufficient condition for K\"{o}nig-Egerv\'{a}ry graphs.

\section*{Acknowledgments}
The authors are grateful to anonymous referees for valuable comments and helpful suggests that greatly improved the presentation of this paper.

\bibliographystyle{abbrv}
\bibliography{reference}
\end{document}